\documentclass[11pt]{llncs}
\usepackage{latexsym}
\usepackage{amssymb}
\usepackage{amsmath,amstext}
 \usepackage[pdftex]{graphicx}
\usepackage{owna4}

\makeatletter
\def\@fnsymbol#1{\ensuremath{\ifcase#1\or\star\or \dagger\or{\star\star}\or
   {\star\star\star}\or \ddagger\or
   \mathchar "278\or \mathchar "27B\or \|\or **\or \dagger\dagger
   \or \ddagger\ddagger \else\@ctrerr\fi}}
\makeatother

\makeatletter
\def~{\ifmmode\;\else\penalty\@M\ \fi}
\def\@setmcodes#1#2#3{{\count0=#1 \count1=#3
  \loop \global\mathcode\count0=\count1 \ifnum \count0<#2
  \advance\count0 by1 \advance\count1 by1 \repeat}}
\DeclareSymbolFont{italic}{OT1}{\rmdefault}{m}{it}
\let\mathit\undefined
\DeclareSymbolFontAlphabet{\mathit}{italic}
\edef\@tempa{\hexnumber@\symitalic}
\@setmcodes{`A}{`Z}{"7\@tempa41}
\@setmcodes{`a}{`z}{"7\@tempa61}
\makeatother
\newdimen\asmindent     
\asmindent=\parindent
\newcount\asmi
\def\inc{\global\advance\asmi by 1}
\def\dec{\global\advance\asmi by-1}
\def\nl{{}$\par\hangindent\asmi em
  \noindent\kern\asmi em\ignorespaces$} 
\def\asmskip{{}$\par\smallskip\hangindent\asmi em
  \noindent\kern\asmi em\ignorespaces$} 
\def\asm{\global\asmi=0 
 \def\+{\inc\nl}
 \def\-{\dec\nl}
 \def\\{\nl}
 \begin{trivlist}\item[]\leftskip=\asmindent\relax$}
\def\endasm{$\end{trivlist}}
\def\asmarray{\begin{array}[t]{@{}l@{\;}l@{\;}l@{}}}
\def\endasmarray{\end{array}}
\newcount\asmii
\def\subasm{\vtop\bgroup\asmii=0\normalbaselines
 \def\nl##1{$\egroup\advance\asmii by##1\relax\hbox\bgroup\hskip\asmii em$}
 \def\\{\nl{0}}
 \def\+{\nl{1}}
 \def\-{\nl{-1}}
 \hbox\bgroup\hskip\asmii em$}
\def\endsubasm{$\egroup\egroup}
\def\ASM#1{\hbox{\sc#1}}        

\def\AND     {\mathrel{\mathbf{and}}}

\def\CHOOSE  {\mathrel{\mathbf{choose}}}

\def\DO      {\mathrel{\mathbf{do}}}
\def\ELSE    {\mathrel{\mathbf{else}}}

\def\FORALL  {\mathrel{\mathbf{forall}}}
\def\FORSOME  {\mathrel{\mathbf{forsome}}}

\def\IF      {\mathrel{\mathbf{if}}}

\def\IN      {\mathrel{\mathbf{in}}}
\def\LET     {\mathrel{\mathbf{let}}}

\def\OR      {\mathrel{\mathbf{or}}}
\def\PAR     {\mathrel{\mathbf{par}}}
\def\SEQ     {\mathrel{\mathbf{seq}}}
\def\SKIP    {\mathrel{\mathbf{skip}}}
\def\THEN    {\mathrel{\mathbf{then}}}
\def\WHERE   {\mathrel{\mathbf{where}}}

\def\WITH    {\mathrel{\mathbf{with}}}

\def\SEQ    {\mathrel{\mathbf{seq}}}

\makeatletter
\def\enumerate{%
  \ifnum \@enumdepth >\thr@@\@toodeep\else
    \advance\@enumdepth\@ne
    \edef\@enumctr{enum\romannumeral\the\@enumdepth}%
      \expandafter
      \list
        \csname label\@enumctr\endcsname
        {\usecounter\@enumctr\def\makelabel##1{\hss\llap{##1}}
         \itemsep 0pt\parskip 0pt\parsep 0pt\topsep\smallskipamount}%
  \fi}
\def\itemize{%
  \ifnum \@itemdepth >\thr@@\@toodeep\else
    \advance\@itemdepth\@ne
    \edef\@itemitem{labelitem\romannumeral\the\@itemdepth}%
    \expandafter
    \list
      \csname\@itemitem\endcsname
      {\def\makelabel##1{\hss\llap{##1}}
       \itemsep 0pt\parskip 0pt\parsep 0pt\topsep\smallskipamount}%
  \fi}
\makeatother
\begin{document}

\title{Specifying Transaction Control \\to Serialize Concurrent
  Program Executions\thanks{The research reported in this paper results
    from the project \textit{Behavioural Theory and Logics for
      Distributed Adaptive Systems} supported by the \textbf{Austrian
      Science Fund (FWF): [P26452-N15]}.}\thanks{The final publication is available at Springer via https://doi.org/10.1007/978-3-662-43652-3\_13.}}

\author{Egon B{\"o}rger\inst{1} and Klaus-Dieter Schewe \inst{2}}
\institute{Universit\`{a} di Pisa, Dipartimento di Informatica,
           I-56125 Pisa, Italy
           \email{boerger@di.unipi.it} \and 
           Software Competence Centre Hagenberg, A-4232 Hagenberg, Austria
           \email{klaus-dieter.schewe@scch.at}
           }

\maketitle 

\begin{abstract}
We define a programming language independent transaction controller
and an operator which when applied to concurrent programs with shared
locations turns their behavior with respect to some abstract
termination criterion into a transactional behavior. We prove the
correctness property that concurrent runs under the transaction
controller are serialisable. We specify the transaction controller
$\ASM{TaCtl}$ and the operator $TA$ in terms of Abstract State
Machines. This makes $\ASM{TaCtl}$ applicable to a wide range of
programs and in particular provides the possibility to use it as a
plug-in when specifying concurrent system components in terms of
Abstract State Machines.
\end{abstract}

\section{Introduction}

This paper is about the use of transactions as a common means to
control concurrent access of programs to shared locations and to avoid
that values stored at these locations are changed almost randomly. A
{\em transaction controller} interacts with concurrently running
programs (read: sequential components of an asynchronous system) to
control whether access to a shared location can be granted or not,
thus ensuring a certain form of consistency for these locations. A
commonly accepted consistency criterion is that the joint behavior of
all transactions (read: programs running under transactional control)
with respect to the shared locations is equivalent to a serial
execution of those programs. Serialisability guarantees that each
transaction can be specified independently from the transaction
controller, as if it had exclusive access to the shared locations.

It is expensive and cumbersome to specify transactional behavior and
prove its correctness again and again for components of the great
number of concurrent systems. Our goal is to define once and for all
an abstract (i.e. programming language independent) transaction
controller $\ASM{TaCtl}$ which can simply be ``plugged in'' to turn
the behavior of concurrent programs (read: components~$M$ of any given
asynchronous system $\cal M$) into a transactional one. This involves
to also define an operator~$TA(M,\ASM{TaCtl})$ which forces the
programs~$M$ to listen to the controller $\ASM{TaCtl}$ when trying to
access shared locations.

For the sake of generality we define the operator and the controller
in terms of Abstract State Machines (ASMs) which can be read and
understood as pseudo-code so that $\ASM{TaCtl}$ and the operator $TA$
can be applied to code written in any programming language (to be
precise: whose programs come with a notion of single step, the level
where our controller imposes shared memory access constraints to
guarantee transactional code behavior). On the other side, the precise
semantics underlying ASMs (for which we refer the reader
to~\cite{BoeSta03}) allows us to mathematically prove the correctness
of our controller and operator.

We concentrate here on transaction controllers that employ locking
strategies such as the common two-phase locking protocol (2PL). That
is, each transaction first has to acquire a (read- or write-) lock for
a shared location, before the access is granted. Locks are released
after the transaction has successfully committed and no more access to
the shared locations is necessary. There are of course other
approaches to transaction handling, see
e.g. \cite{elmasri:2006,gray:1993,kirchberg:2009,schewe:actac2000} and
the extensive literature there covering classical transaction control
for flat transactions, timestamp-based, optimistic and hybrid
transaction control protocols, as well as non-flat transaction models
such as sagas and multi-level transactions.

We define $\ASM{TaCtl}$ and the operator $TA$ in
Sect.~\ref{sect:TAoperator} and the $\ASM{TaCtl}$ components in
Sect.~\ref{sect:TaCtl}. In Sect.~\ref{sect:Thm} we prove the
correctness of these definitions.

\section{The Transaction Operator $TA(M,\ASM{TaCtl}$)}
\label{sect:TAoperator}

As explained above, a transaction controller performs the lock
handling, the deadlock detection and handling, the recovery mechanism
(for partial recovery) and the commit of single machines. Thus we
define it as consisting of four components specified in
Sect.~\ref{sect:TaCtl}.

\begin{asm}
\ASM{TaCtl}=\+
\ASM{LockHandler} \\
\ASM{DeadlockHandler} \\
\ASM{Recovery} \\
\ASM{Commit} 
\end{asm}

The operator~$TA(M,\ASM{TaCtl})$ transforms the components~$M$ of any
concurrent system (asynchronous ASM) ${\cal M} = (M_i)_{i \in
  I}$ into components of a concurrent system $TA({\cal M},\ASM{TaCtl})$ where
each $TA(M_i,\ASM{TaCtl})$ runs as transaction under the control
of~$\ASM{TaCtl}$:
\begin{asm}
TA({\cal M},\ASM{TaCtl})= ((TA(M_i,\ASM{TaCtl}))_{i \in I},\ASM{TaCtl})
\end{asm}

$\ASM{TaCtl}$ keeps a dynamic set $TransAct$ of those machines~$M$
whose runs it currently has to supervise to perform in a transactional
manner until~$M$ has $Terminated$ its transactional behavior (so that
it can $\ASM{Commit}$ it).\footnote{In this paper we deliberately keep
  the termination criterion abstract so that it can be refined in
  different ways for different transaction instances.}  To turn the
behavior of a machine~$M$ into a transactional one, first of all~$M$
has to register itself with the controller $\ASM{TaCtl}$, read: to be
inserted into the set of currently to be handled $TransAct$ions. To
$\ASM{Undo}$ as part of a recovery some steps~$M$ made already during
the given transactional run segment of~$M$, a last-in first-out queue
$history(M)$ is needed which keeps track of the states the
transactional run goes through; when~$M$ enters the set $TransAct$ the
$history(M)$ has to be initialized (to the empty queue).

The crucial transactional feature is that each non private
(i.e. shared or monitored or output) location~$l$ a machine~$M$ needs
to read or write for performing a step has to be $LockedBy(M)$ for
this purpose;~$M$ tries to obtain such locks by calling the
$\ASM{LockHandler}$. In case no $newLocks$ are needed by~$M$ in its
$currState$ or the needed $newLocks$ can be $Granted$ by the
$\ASM{LockHandler}$, $M$ performs its next step; in addition, for a
possible future recovery, the machine has to $\ASM{Record}$ in its
$history(M)$ the current values of those locations which are (possibly
over-) written by this~$M$-step together with the obtained
$newLocks$. Then~$M$ continues its transactional behavior until it is
$Terminated$. In case the needed $newLocks$ are $Refused$, namely
because another machine~$N$ in $TransAct$ for some needed~$l$ has
$W\mbox{-}Locked(l,N)$ or (in case~$M$ wants a W-(rite)Lock) has
$R\mbox{-}Locked(l,N)$, $M$ has to $Wait$ for~$N$; in fact it
continues its transactional behavior by calling again the
$\ASM{LockHandler}$ for the needed $newLocks$---until the needed
locked locations are unlocked when~$N$'s transactional behavior is
$\ASM{Commit}$ed, whereafter a new request for these locks this time
may be $Granted$ to~$M$.\footnote{As suggested by a reviewer, a
  refinement (in fact a desirable optimization) consists in replacing
  such a waiting cycle by suspending~$M$ until the needed locks are
  released. Such a refinement can be obtained in various ways, a
  simple one consisting in letting~$M$ simply stay in $waitForLocks$
  until the $newLocks$ $CanBeGranted$ and refining $\ASM{LockHandler}$
  to only choose pairs $(M,L)\in LockRequest$ where it can
  $\ASM{GrantRequestedLocks}(M,L)$ and doing nothing otherwise
  (i.e. defining $\ASM{RefuseRequestedLocks}(M,L)=~\SKIP$). See
  Sect.~\ref{sect:TaCtl}.}

As a consequence deadlocks may occur, namely when a cycle occurs in
the transitive closure $Wait^*$ of the $Wait$ relation.  To resolve
such deadlocks the $\ASM{DeadlockHandler}$ component of $\ASM{TaCtl}$
chooses some machines as $Victim$s for a recovery.\footnote{To
  simplify the serializability proof in Sect.\ref{sect:TaCtl} and
  without loss of generality we define a reaction of machines~$M$ to
  their victimization only when they are in $ctl\_state(M)=~$TA-$ctl$
  (not in $ctl\_state(M)=waitForLocks$). This is to guarantee that no
  locks are $Granted$ to a machine as long as it does
  $waitForRecovery$.} After a victimized machine~$M$ is $Recovered$ by
the $\ASM {Recovery}$ component of $\ASM{TaCtl}$, so that~$M$ can exit
its $waitForRecovery$ state, it continues its transactional behavior.
 
This explains the following definition of $TA(M,\ASM{TaCtl})$ as a
control state ASM, i.e. an ASM with a top level Finite State Machine
control structure. We formulate it by the flowchart diagram of
Fig.~\ref{fig:TA(M,C)}, which has a precise control state ASM
semantics (see the definition in~\cite[Ch.2.2.6]{BoeSta03}). The
components for the recovery feature are highlighted in the flowchart
by a colouring that differs from that of the other components. The
macros which appear in Fig.~\ref{fig:TA(M,C)} and the components of
$\ASM{TaCtl}$ are defined below.

\begin{figure}[htb]
  \begin{center}
  \includegraphics[width=0.9\textwidth]{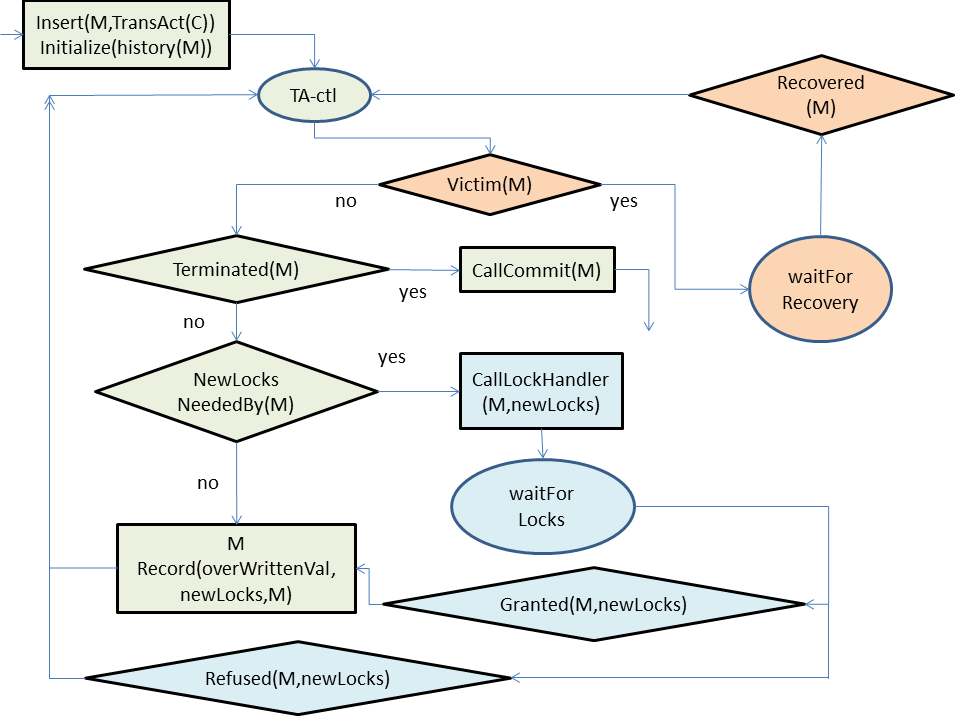}
  \end{center}
  \caption{TA(M,C)}
  \label{fig:TA(M,C)}
\end{figure}

The predicate $NewLocksNeededBy(M)$ holds if in the current state
of~$M$ at least one of two cases
happens:\footnote{See~\cite[Ch.2.2.3]{BoeSta03} for the classification
  of locations and functions.} either $M$ to perform its step in this
state reads some shared or monitored location which is not yet
$LockedBy(M)$ or~$M$ writes some shared or output location which is
not yet $LockedBy(M)$ for writing. A location can be $LockedBy(M)$ for
reading ($R\mbox{-}Locked(l,M)$) or for writing
($W\mbox{-}Locked(l,M)$).  Formally:
\begin{asm}
NewLocksNeededBy(M)=\+
   newLocks(M,currState(M))\footnote{
   For layout reasons we omit in Fig.\ref{fig:TA(M,C)} the arguments of the
   functions $newLocks$ and $overWrittenVal$.}\not = (\emptyset,\emptyset) \-
newLocks(M,currState(M))\footnote{By the second argument 
  $currState(M)$ of $newLocks$ (and below of $overWrittenVal$) we indicate that this 
  function of~$M$ is a dynamic function which is evaluated in each 
  state of~$M$, namely by computing in this state the sets $ReadLoc(M)$ and 
  $WriteLoc(M)$; see Sect.~\ref{sect:Thm} for the detailed definition.}
                 =(R\mbox{-}Loc,W\mbox{-}Loc)\+
   \WHERE \\
   R\mbox{-}Loc = ReadLoc(M,currState(M)) \cap (SharedLoc(M) \cup MonitoredLoc(M))\+
           ~~~~~~~~~~~~~~~~\cap \overline{LockedBy(M)}
                 \footnote{By $\overline{X}$ we denote the complement of~$X$.}\-
   W\mbox{-}Loc = WriteLoc(M,currState(M)) \cap (SharedLoc(M) \cup OutputLoc(M))\+
           ~~~~~~~~~~~~~~~~\cap \overline{W\mbox{-}LockedBy(M)}\dec\-
LockedBy(M)=  \{l \mid R\mbox{-}Locked(l,M) \OR W\mbox{-}Locked(l,M)\}\\
 W\mbox{-}LockedBy(M)=  \{l \mid  W\mbox{-}Locked(l,M)\}
\end{asm}

The $overWrittenVal$ues are the $currState(M)$-values (retrieved by
the $eval$-function) of those shared or output locations $(f,args)$
which are written by~$M$ in its $currState(M)$. To $\ASM{Record}$ the
set of these values together with the obtained $newLocks$ means to
append the pair of these two sets to the $history$ queue of~$M$ from
where upon recovery the values and the locks can be retrieved.

\begin{asm}
overWrittenVal(M,currState(M))= ~ \{((f,args),val) \mid \+
  (f,args) \in WriteLoc(M,currState(M))\cap (SharedLoc(M) \cup OutputLoc(M))\\ 
  \AND val = eval(f(args),currState(M))\} \-
\ASM{Record}(valSet,lockSet,M) = ~ \ASM{Append}((valSet,lockSet),history(M))
\end{asm}

To $\ASM{CallLockHandler}$ for the $newLocks$ requested by~$M$ in its
$currState(M)$ means to $\ASM{Insert}(M,newLocks)$ into the
$\ASM{LockHandler}$'s set of to be handled $LockRequest$s. Similarly
we let $\ASM{CallCommit(M)}$ stand for insertion of~$M$ into a set
$CommitRequest$ of the $\ASM{Commit}$ component.

\begin{asm}
\ASM{CallLockHandler}(M,L)=~\ASM{Insert}((M,L),LockRequest)\\
\ASM{CallCommit}(M)=~\ASM{Insert}(M,CommitRequest)
\end{asm}

\section{The Transaction Controller Components}
\label{sect:TaCtl}

A $\ASM{CallCommit(M)}$ by machine~$M$ enables the $\ASM{Commit}$
component. Using the $\CHOOSE$ operator we leave the order in which
the $CommitRequest$s are handled refinable by different instantiations
of $\ASM{TaCtl}$.

$\ASM{Commit}$ing~$M$ means to $\ASM{Unlock}$ all locations~$l$ that
are $LockedBy(M)$. Note that each lock obtained by~$M$ remains
with~$M$ until the end of~$M$'s transactional behavior. Since~$M$
performs a $\ASM{CallCommit(M)}$ when it has $Terminated$ its
transactional computation, nothing more has to be done to
$\ASM{Commit}$ $M$ besides deleting~$M$ from the sets of
$CommitRequest$s and still to be handled $TransAct$ions.\footnote{We
  omit clearing the $history(M)$ queue since it is initialized
  when~$M$ is inserted into $TransAct(\ASM{TaCtl})$.}

Note that the locations $R\mbox{-}Locked(l,M)$ and
$W\mbox{-}Locked(l,M)$ are shared by the $\ASM{Commit}$,
$\ASM{LockHandler}$ and $\ASM{Recovery}$ components, but these
components never have the same~$M$ simultaneously in their request
resp. $Victim$ set since when machine~$M$ has performed a
$\ASM{CallCommit(M)}$, it has $Terminated$ its transactional
computation and does not participate any more in any $(M,L) \in
LockRequest$ or $Victim$ization.

\begin{asm}
\ASM{Commit} =\+
  \IF CommitRequest \not = \emptyset \THEN \+
     \CHOOSE M \in CommitRequest ~\ASM{Commit}(M) \-
  \WHERE \+
    \ASM{Commit}(M) = \+
       \FORALL l \in LockedBy(M) ~~ \ASM{Unlock}(l,M) \\
       \ASM{Delete}(M,CommitRequest)\\
       \ASM{Delete}(M,TransAct)\-
    \ASM{Unlock}(l,M)= \+
      \IF R\mbox{-}Locked(l,M) \THEN  R\mbox{-}Locked(l,M):=false  \\
      \IF W\mbox{-}Locked(l,M) \THEN  W\mbox{-}Locked(l,M):=false 
\end{asm}

As for $\ASM{Commit}$ also for the $\ASM{LockHandler}$ we use the
$\CHOOSE$ operator to leave the order in which the $LockRequest$s are
handled refinable by different instantiations of $\ASM{TaCtl}$.

The strategy we adopt for lock handling is to refuse all locks for
locations requested by~$M$ if at least one of the following two cases
happens:

\begin{itemize}
\item some of the requested locations is $W\mbox{-}Locked$ by another
  transactional machine~$N \in TransAct$,
\item some of the requested locations is a $WriteLoc$ation that is
  $R\mbox{-}Locked$ by another transactional machine~$N \in TransAct$.
\end{itemize}

 This definition implies that multiple transactions may
 simultaneoulsy have a $R\mbox{-}Lock$ on some location. It is specified below by the predicate
 $CannotBeGranted$. 

To $\ASM{RefuseRequestedLocks}$ it suffices to set the communication
 interface $Refused$ of $TA(M,\ASM{TaCtl})$; this makes~$M$ $Wait$ for
 each location~$l$ that is $ W\mbox{-}Locked(l,N)$ and for each
 $WriteLoc$ation that is $R\mbox{-}Locked(l,N)$ by some other transactional
 component machine $N \in TransAct$.

\begin{asm}
\ASM{LockHandler} =\+
\IF LockRequest \not = \emptyset \THEN \+
   \CHOOSE (M,L) \in LockRequest  \+
       \ASM{HandleLockRequest}(M,L)\dec\-
\WHERE \+
 \ASM{HandleLockRequest}(M,L) = \+
      \IF CannotBeGranted(M,L) \+
         \THEN  ~ \ASM{RefuseRequestedLocks}(M,L)\\
         \ELSE  ~ \ASM{GrantRequestedLocks}(M,L) \-
      \ASM{Delete}((M,L),LockRequest)\-         
 CannotBeGranted(M,L)=\+
    \LET L=(R\mbox{-}Loc,W\mbox{-}Loc), Loc= R\mbox{-}Loc \cup W\mbox{-}Loc\+
       \FORSOME l \in Loc ~~\FORSOME N \in TransAct \setminus \{M\}\+ 
          W\mbox{-}Locked(l,N) \OR \+
               (l \in W\mbox{-}Loc \AND R\mbox{-}Locked(l,N))\dec\dec\-
 \ASM{RefuseRequestedLocks}(M,L) =  (Refused(M,L):=true)\\
 \ASM{GrantRequestedLocks}(M,L)=\+
    \LET L=(R\mbox{-}Loc,W\mbox{-}Loc)\+
       \FORALL l \in R\mbox{-}Loc ~~(R\mbox{-}Locked(l,M):=true)\\    
       \FORALL l \in W\mbox{-}Loc ~~(W\mbox{-}Locked(l,M):=true)\- 
    Granted(M,L):=true   
\end{asm}

A $Deadlock$ originates if two machines are in a $Wait$ cycle,
otherwise stated if for some (not yet $Victim$ized) machine~$M$ the
pair $(M,M)$ is in the transitive (not reflexive) closure $Wait^*$ of
$Wait$. In this case the $\ASM{DeadlockHandler}$ selects for recovery
a (typically minimal) subset of $Deadlocked$ transactions
$toResolve$---they are $Victim$ized to $waitForRecovery$, in which
mode (control state) they are backtracked until they become
$Recovered$. The selection criteria are intrinsically specific for
particular transaction controllers, driving a usually rather complex
selection algorithm in terms of number of conflict partners,
priorities, waiting time, etc.  In this paper we leave their
specification for $\ASM{TaCtl}$ abstract (read: refinable in different
directions) by using the $\CHOOSE$ operator.

\begin{asm}
\ASM{DeadlockHandler} =\+
\IF Deadlocked \cap \overline{Victim} \not = \emptyset \THEN 
                  \mbox{ // there is a Wait cycle}\+
    \CHOOSE toResolve \subseteq Deadlocked \cap \overline{Victim} \+  
      \FORALL M \in toResolve~ Victim(M):=true\dec\-  
\WHERE \+
Deadlocked = \{M \mid (M,M) \in M^*\}\\
M^*= \mbox{ TransitiveClosure}(Wait)\\
Wait(M,N) = ~ \FORSOME l ~ Wait(M,l,N) \\
Wait(M,l,N) = \+
     l \in newLocks(M,currState(M)) \AND   N \in TransAct \setminus \{M\}  
            \AND \\       
     W\mbox{-}Locked(l,N) \OR (l \in W\mbox{-}Loc \AND R\mbox{-}Locked(l,N))\+
           \WHERE newLocks(M,currState(M)) =(R\mbox{-}Loc,W\mbox{-}Loc)\dec\-
\end{asm}

Also for the $\ASM{Recovery}$ component we use the $\CHOOSE$ operator
to leave the order in which the $Victim$s are chosen for recovery
refinable by different instantiations of $\ASM{TaCtl}$. To be
$Recovered$ a machine~$M$ is backtracked by $\ASM{Undo}(M)$ steps
until $M$ is not $Deadlocked$ any more, in which case it is
deleted from the set of $Victim$s, so that be definition it is
$Recovered$. This happens at the latest when $history(M)$ has become
empty.

\begin{asm}
\ASM{Recovery} =\+
   \IF Victim \not = \emptyset \THEN \+
      \CHOOSE M \in Victim ~\ASM{TryToRecover}(M)\dec \-          
\WHERE \+
\ASM{TryToRecover}(M) = \+
     \IF M \not \in Deadlocked \THEN Victim(M):=false \+
                \ELSE ~ \ASM{Undo}(M)\dec\-
Recovered = \+
   \{M \mid ctl\mbox{-}state(M)=waitForRecovery \AND M \not \in Victim\}\-
\ASM{Undo}(M)= \+
  \LET (ValSet,LockSet) = youngest(history(M))\+
     \ASM{Restore}(ValSet)\\
     \ASM{Release}(LockSet)\\
     \ASM{Delete}((ValSet,LockSet),history(M))\dec\-
\WHERE \+
   \ASM{Restore}(V) =\+
       \FORALL ((f,args),v) \in V ~ f(args):=v \-
   \ASM{Release}(L)= \+
      \LET L = (R\mbox{-}Loc,W\mbox{-}Loc)\+
        \FORALL l \in Loc=R\mbox{-}Loc \cup W\mbox{-}Loc ~ \ASM{Unlock}(l,M)
\end{asm}

Note that in our description of the \ASM{DeadlockHandler} and the
(partial) \ASM{Recovery} we deliberately left the strategy for victim
seclection and $\ASM{Undo}$ abstract leaving fairness considerations to be
discussed elsewhere. It is clear that if always the same victim is
selected for partial recovery, the same deadlocks may be created again
and again. However, it is well known that fairness can be achieved by
choosing an appropriate victim selection strategy.

\section{Correctness Theorem}
\label{sect:Thm}

In this section we show the desired correctness property: if all
monitored or shared locations of any $M_i$ are output or controlled
locations of some other $M_j$ and all output locations of any $M_i$
are monitored or shared locations of some other $M_j$ (closed system
assumption)\footnote{This assumption means that the environment is
  assumed to be one of the component machines.}, each run
of $TA({\cal M},\ASM{TaCtl})$ is equivalent to a serialization of the
terminating $M_i$-runs, namely the $M_{i_{1}}$-run followed by the
$M_{i_{2}}$-run etc., where $M_{i_{j}}$ is the $j$-th machine of $\cal
M$ which performs a commit in the $TA({\cal M},\ASM{TaCtl})$ run. To
simplify the exposition (i.e. the formulation of statement and proof
of the theorem) we only consider machine steps which take place under
the transaction control, in other words we abstract from any
step~$M_i$ makes before being $\ASM{Insert}$ed into or after being
$\ASM{Delet}$ed from the set $TransAct$ of machines which currently
run under the control of $\ASM{TaCtl}$.

First of all we have to make precise what a {\em serial} multi-agent
ASM run is and what {\em equivalence} of $TA({\cal M},\ASM{TaCtl})$
runs means in the general multi-agent ASM framework.

\subsubsection{Definition of run equivalence.}

Let $S_0, S_1, S_2, \dots$ be a (finite or infinite) run of the system $TA({\cal
  M},\ASM{TaCtl})$. In general we may assume that \ASM{TaCtl} runs
  forever, whereas each machine $M \in \mathcal{M}$ running as
  transaction will be terminated at some time -- at least after commit
  $M$ will only change values of non-shared and non-output
  locations\footnote{It is possible that one ASM $M$ enters several
    times as a transaction controlled by \ASM{TaCtl}. However, in this
    case each of these registrations will be counted as a separate
    transaction, i.e. as different ASMs in $\mathcal{M}$.}. For $i =
  0,1,2,\dots$ let $\Delta_i$ denote the unique, consistent update set
  defining the transition from $S_i$ to $S_{i+1}$. By definition of $TA({\cal
    M},\ASM{TaCtl})$ the update set is the union of the update sets of
  the agents executing $M \in \mathcal{M}$ resp. $\ASM{TaCtl}$:
\[ \Delta_i = \bigcup\limits_{M \in \mathcal{M}} \Delta_i(M) 
      \cup \Delta_i(\ASM{TaCtl}). \]

\noindent $\Delta_i(M)$ contains the updates defined by the ASM
$TA(M,\ASM{TaCtl})$ in state $S_i$\footnote{We use the shorthand
  notation $\Delta_i(M)$ to denote $\Delta_i(TA(M,\ASM{TaCtl}))$; in
  other words we speak about steps and updates of~$M$ also when they
  really are done by~$TA(M,\ASM{TaCtl})$. Mainly this is about
  transitions between the control states, namely TA-$ctl$,
  $waitForLocks$, $waitForRecovery$ (see Fig.\ref{fig:TA(M,C)}), which
  are performed during the run of~$M$ under the control of the
  transaction controller $\ASM{TaCtl}$. When we want to name an
  original update of~$M$ (not one of the updates of $ctl\_state(M)$ or
  of the $\ASM{Record}$ component) we call it a proper $M$-update.}
and $\Delta_i(\ASM{TaCtl})$ contains the updates by the transaction
controller in this state. The sequence of update sets $\Delta_0(M)$,
$\Delta_1(M)$, $\Delta_2(M)$, \dots will be called the {\em schedule}
of $M$ (for the given transactional run).

To generalise for transactional ASM runs the equivalence of
transaction schedules known from database systems
\cite[p.621ff.]{elmasri:2006} we now define two {\em cleansing
  operations} for ASM schedules. By the first one (i) we eliminate all
(in particular unsuccessful-lock-request) computation segments which
are without proper $M$-updates; by the second one (ii) we eliminate
all $M$-steps which are related to a later $\ASM{Undo}(M)$ step by the
$\ASM{Recovery}$ component:

\begin{enumerate}

\item Delete from the schedule of $M$ each $\Delta_i(M)$ where one
  of the following two properties holds:

    \begin{itemize}
    \item $\Delta_i(M)=\emptyset$ ($M$ contributes no update to $S_i$),
   \item $\Delta_i(M)$ belongs to a step of an $M$-computation segment
     where~$M$ in its $ctl\_state(M)=$ TA-$ctl$ does
     $\ASM{CallLockHandler}(M,newLocks)$ and in its next step moves
     from control-state $waitForLocks$ back to control state TA$-ctl$, because the
     $\ASM{LockHandler}$ refused new locks by $Refused(M,newLocks)$.\footnote{Note that by
       eliminating this $\ASM{CallLockHandler}(M,L)$ step also the
       corresponding $\ASM{LockHandler}$ step
       $\ASM{HandleLockRequest}(M,L)$ disappears in the run.}
     \end{itemize}
In such computation steps~$M$ makes no proper update.

\item Repeat choosing from the schedule of $M$ a pair
  $\Delta_j(M)$ with later $\Delta_{j'}(M)$ ($j<j'$) which belong
  to the first resp. second of two consecutive $M$-Recovery steps
  defined as follows:

  \begin{itemize}
    \item a (say $M$-RecoveryEntry) step whereby~$M$ in state $S_j$
      moves from control-state TA-$ctl$ to $waitForRecovery$, because it became a
      $Victim$,
    \item the next $M$-step (say $M$-RecoveryExit) whereby~$M$ in
      state $S_{j'}$ moves back to control state TA-$ctl$ because it has
      been $Recovered$.
  \end{itemize} 
In these two $M$-Recovery steps~$M$ makes no proper update. Delete:
\begin{enumerate}
\item $\Delta_j(M)$ and $\Delta_{j'}(M)$, 
\item the $((Victim,M),true)$ update from the corresponding
  $\Delta_t(\ASM{TaCtl})$ ($t< j$) which in state $S_j$ triggered the
  $M$-RecoveryEntry,
\item $\ASM{TryToRecover}(M)$-updates in any update set
  $\Delta_{i+k}(\ASM{TaCtl})$ between
  the considered $M$-RecoveryEntry and $M$-RecoveryExit step
  ($i<j<i+k<j'$),
\item each $\Delta_{i'}(M)$ belonging to the $M$-computation segment
  from TA-$ctl$ back to TA-$ctl$ which contains the
  proper $M$-step in $S_i$ that is $\ASM{UNDO}$ne in $S_{i+k}$ by the
  considered $\ASM{TryToRecover}(M)$ step; besides 
     control state and $\ASM{Record}$ updates these $\Delta_{i'}(M)$
     contain updates $(\ell,v)$ with $\ell =
     (f,(val_{S_i}(t_1),\dots,val_{S_i}(t_n)))$ where the
     corresponding $\ASM{Undo}$ updates are
     $(\ell,val_{S_i}(f(t_1,\dots,t_n))) \in
     \Delta_{i+k}(\ASM{TaCtl})$,
 \item the $\ASM{HandleLockRequest}(M,newLocks)$-updates in
   $\Delta_{l\prime}(\ASM{TaCtl})$ corresponding to $M$'s
   $\ASM{CallLockHandler}$ step (if any: in case $newLocks$ are needed
   for the proper $M$-step in $S_i$) in state $S_l$ ($l<l^\prime<i$).
 \end{enumerate}

\end{enumerate}

The sequence $\Delta_{i_1}(M), \Delta_{i_2}(M), \dots$ with $i_1 < i_2
< \dots$ resulting from the application of the two cleansing
operations as long as possible -- note that confluence is obvious, so
the sequence is uniquely defined -- will be called the {\em cleansed
  schedule} of $M$ (for the given run).

Before defining the equivalence of transactional ASM runs we
remark that $TA({\cal M},\ASM{TaCtl})$ has indeed several runs, even
for the same initial state $S_0$. This is due to the fact that a lot
of non-determinism is involved in the definition of this ASM. First,
the submachines of \ASM{TaCtl} are non-deterministic:

\begin{itemize}

\item In case several machines $M, M^{\prime} \in \mathcal{M}$ request
  conflicting locks at the same time, the \ASM{LockHandler} can only
  grant the requested locks for one of these machines.

\item Commit requests are executed in random order by the \ASM{Commit} submachine.

\item The submachine \ASM{DeadlockHandler} chooses a set of victims,
  and this selection has been deliberately left abstract.

\item The \ASM{Recovery} submachine chooses in each step a victim $M$,
  for which the last step will be undone by restoring previous values
  at updated locations and releasing corresponding locks.

\end{itemize}

Second, the specification of $TA({\cal M},\ASM{TaCtl})$ leaves
deliberately open, when a machine $M \in \mathcal{M}$ will be started,
i.e., register as a transaction in $TransAct$ to be controlled by
\ASM{TaCtl}. This is in line with the common view that transactions $M
\in \mathcal{M}$ can register at any time to the transaction
controller \ASM{TaCtl} and will remain under its control until they
commit.

\begin{definition}\rm

\ Two runs $S_0,S_1,S_2,\dots$ and
$S_0^{\prime},S_1^{\prime},S_2^{\prime},\dots$ of $TA({\cal
  M},\ASM{TaCtl})$ are {\em equivalent} iff for each $M \in
\mathcal{M}$ the cleansed schedules $\Delta_{i_1}(M), \Delta_{i_2}(M),
\dots$ and $\Delta_{j_1}^\prime(M), \Delta_{j_2}^\prime(M), \dots$ for
the two runs are the same and the read locations and the values read
by~$M$ in $S_{i_k}$ and $S_{j_k}'$ are the same.

\end{definition}

That is, we consider runs to be equivalent, if all 
transactions $M \in \mathcal{M}$ read the same locations and see there
the same values and perform the same updates in the same order
disregarding waiting times and updates that are undone.

\subsubsection{Definition of serializability.} 

Next we have to clarify our generalised notion of a serial run, for
which we concentrate on committed transactions -- transactions that
have not yet committed can still undo their updates, so they must be
left out of consideration\footnote{Alternatively, we could concentrate
  on complete, infinite runs, in which only committed transactions
  occur, as eventually every transaction will commit -- provided that
  fairness can be achieved.}. We need a definition of the read- and
write-locations of $M$ in a state $S$, i.e. $ReadLoc(M,S)$ and
$WriteLoc(M,S)$ as used in the definition of $newLocks(M,S)$.

The definition of $Read/WriteLoc$ depends on the locking level,
whether locks are provided for variables, pages, blocks, etc. To
provide a definite definition, in this paper we give the definition at
the level of abstraction of the locations of the underlying class
$\cal{M}$ of component machines (ASMs)~$M$. Refining this definition
(and that of $newLocks$) appropriately for other locking levels does
not innvalidate the main result of this paper.

We define $ReadLoc(M,S) = ReadLoc(r,S)$, where $r$ is the defining rule of
the ASM $M$, and analogously $WriteLoc(M,S)$ $= WriteLoc(r,S)$. Then we use structural induction according to the
definition of ASM rules in ~\cite[Table 2.2]{BoeSta03}. As an
auxiliary concept we need to define inductively the read and write
locations of terms and formulae. The definitions use an
interpretation~$I$ of free variables which we suppress notationally
(unless otherwise stated) and assume to be given with (as environment
of) the state~$S$. This allows us to write $ReadLoc(M,S)$,
$WriteLoc(M,S)$ instead of $ReadLoc(M,S,I)$, $ReadLoc(M,S,I)$
respectively.

\subsubsection{Read/Write Locations of Terms and Formulae.}

For state~$S$ let~$I$ be the given interpretation of the variables
which may occur freely (in given terms or formulae). We write
$val_S(construct)$ for the evaluation of~$construct$ (a term or a
formula) in state~$S$ (under the given interpretation~$I$ of free
variables).

\begin{asm}
ReadLoc(x,S)= WriteLoc(x,S)= \emptyset \mbox{ for variables }x \\
ReadLoc(f(t_1, \ldots ,  t_n) ,S)=\+
           \{(f,(val_S(t_1), \ldots , val_S(t_n)))\} 
                 ~\cup ~\bigcup_{1 \leq i \leq n}ReadLoc(t_i,S)\-
WriteLoc(f(t_1 , \ldots , t_n),S)=\{(f,(val_S(t_1), \ldots , val_S(t_n)))\}
\end{asm}
Note that logical variables are not locations: they cannot be written
and their values are not stored in a location but in the given
interpretation~$I$ from where they can be retrieved.

We define $WriteLoc(\alpha,S)=\emptyset$ for every formula $\alpha$
because formulae are not locations one could write into.
$ReadLoc(\alpha,S)$ for atomic formulae $P(t_1 , \ldots , t_n )$ has to be
defined as for terms with $P$ playing the same role as a function
symbol~$f$. For propositional formulae one reads the locations of
their subformulae. In the inductive step for quantified formulae
$domain(S)$ denotes the superuniverse of~$S$
minus the Reserve set~\cite[Ch.2.4.4]{BoeSta03} and 
$I_{x}^{d}$ the extension (or modification) of~$I$ where~$x$ is
interpreted by a domain element~$d$.

\begin{asm}
ReadLoc(P(t_1 , \ldots , t_n ),S)= \+
   \{(P,(val_S(t_1), \ldots , val_S(t_n)))\} 
                   ~\cup ~ \bigcup_{1 \leq i \leq n}ReadLoc(t_i,S)\-
ReadLoc(\neg \alpha)= ReadLoc(\alpha)\\
ReadLoc(\alpha_1 \wedge \alpha_2) = ReadLoc(\alpha_1) \cup ReadLoc(\alpha_2)\\
ReadLoc(\forall x \alpha,S,I)= \bigcup_{d \in domain(S)}ReadLoc(\alpha,S,I_{x}^{d})
\end{asm}
\noindent Note that the values of the logical variables are not read
from a location but from the modified state environment function~$I_{x}^{d}$.

\subsubsection{Read/Write Locations of ASM Rules.}

\begin{asm}
ReadLoc(\SKIP,S)=WriteLoc(\SKIP,S)= \emptyset\\
ReadLoc(t_1 := t_2,S)= ReadLoc(t_1,S) \cup ReadLoc(t_2,S)\\
WriteLoc(t_1 := t_2,S)= WriteLoc(t_1,S)\\
ReadLoc(\IF \alpha \THEN r_1 \ELSE r_2,S) = \+
     ReadLoc(\alpha,S) \cup \left\{ 
           \begin{array}{ll}
             ReadLoc(r_1,S) & \IF val_S(\alpha)=true \\
             ReadLoc(r_2,S) & \ELSE 
            \end{array} \right.\-
WriteLoc(\IF \alpha \THEN r_1 \ELSE r_2,S) =
          \left\{ 
           \begin{array}{ll}
             WriteLoc(r_1,S) & \IF val_S(\alpha)=true\\                                     
             WriteLoc(r_2,S) & \ELSE 
           \end{array} \right. \\
ReadLoc(\LET x = t \IN r,S,I)= ReadLoc(t,S,I) \cup ReadLoc(r,S,I_x^{val_S(t)})\\
WriteLoc(\LET x = t \IN r,S,I)= WriteLoc(r,S,I_x^{val_S(t)}) \mbox{ // call by value}\\
ReadLoc(\FORALL x \WITH \alpha \DO r,S,I)=  \+
     ReadLoc(\forall x  \alpha,S,I) ~\cup ~ \bigcup_{a \in range(x,\alpha,S,I)}ReadLoc(r,S,I_x^a)\+
     \WHERE range(x,\alpha,S,I)= \{d \in domain(S) \mid val_{S,I_x^d}(\alpha) = true\}\dec\-
WriteLoc(\FORALL x \WITH \alpha \DO r,S,I)=  
     \bigcup_{a \in range(x,\alpha,S,I)}WriteLoc(r,S,I_x^a)\\
\end{asm}
In the following cases the same scheme applies to read and write
locations:\footnote{In $yields(r_1,S,I,U)$~$U$ denotes the update set
  produced by rule~$r_1$ in state~$S$ under~$I$.}
\begin{asm}
Read[Write]Loc(r_1 \PAR r_2,S)= \+
           Read[Write]Loc(r_1,S) \cup Read[Write]Loc(r_2,S)\-
Read[Write]Loc(r(t_1, \ldots, t_n),S) = Read[Write]Loc(P(x_1/t_1,\ldots,x_n/t_n),S)  \+
     \WHERE r(x_1,\ldots,x_n)=P \mbox{ // call by reference}\-
Read[Write]Loc(r_1 \SEQ r_2,S,I) = Read[Write]Loc(r_1,S,I) \cup \+
    \left\{ 
           \begin{array}{ll}
            Read[Write]Loc(r_2,S+U,I) & \IF
            yields(r_1,S,I,U) \AND Consistent(U)\\  
            \emptyset & \ELSE
     \end{array} \right.
\end{asm}
For $\CHOOSE$ rules we have to define the read and write locations
simultaneously to guarantee that the same instance satisfying the
selection condition is chosen for defining the read and write
locations of the rule body~$r$:
\begin{asm}
\IF range(x,\alpha,S,I)= \emptyset \THEN \+
   ReadLoc(\CHOOSE x \WITH \alpha \DO r,S,I)= ReadLoc(\exists x \alpha,S,I)\\
   WriteLoc(\CHOOSE x \WITH \alpha \DO r,S,I)=\emptyset \mbox{ // empty action}\-
\ELSE ~\CHOOSE a \in range(x,\alpha,S,I)\+
   ReadLoc(\CHOOSE x \WITH \alpha \DO r,S,I)= \+
          ReadLoc(\exists x \alpha,S,I) 
                  \cup ReadLoc(r,S,I_x^a)\-
   WriteLoc(\CHOOSE x \WITH \alpha \DO r,S,I)= WriteLoc(r,S,I_x^a)
\end{asm}

We say that~$M$ has or is committed (in state~$S_i$, denoted
$Committed(M,S_i)$) if step $\ASM{Commit}(M)$ has been performed (in
state~$S_i$).

\begin{definition}\rm

\ A run of $TA({\cal M},\ASM{TaCtl})$ is {\em serial} iff there is a
total order $<$ on $\mathcal{M}$ such that the following two
conditions are satisfied:

\begin{enumerate}

\item If in a state $M$ has committed, but $M^\prime$ has not, then
  $M < M^\prime$ holds.

\item If $M$ has committed in state $S_i$ and $M < M^\prime$ holds,
  then the cleansed schedule $\Delta_{j_1}(M^\prime)$,
  $\Delta_{j_2}(M^\prime), \dots$ of $M^\prime$ satisfies $i < j_1$.

\end{enumerate}

\end{definition}

That is, in a serial run all committed transactions are executed
in a total order and are followed by the updates of transactions that did not
yet commit.

\begin{definition}\rm

\ A run of $TA({\cal M},\ASM{TaCtl})$ is {\em serialisable} iff it is
equivalent to a serial run of $TA({\cal
  M},\ASM{TaCtl})$.\footnote{Modulo the fact that ASM steps permit
  simultaneous updates of multiple locations, this definition of
  serializability is equivalent to Lamport's sequential consistency
  concept~\cite{Lamport79}.}

\end{definition}

\begin{theorem}

\ Each run of $TA({\cal M},\ASM{TaCtl})$ is serialisable.

\end{theorem}

\begin{proof}
\ Let $S_0,S_1,S_2,\dots$ be a run of $TA({\cal M},\ASM{TaCtl})$. To
construct an equivalent serial run let $M_1 \in \mathcal{M}$ be a
machine that commits first in this run, i.e. $Committed(M,S_i)$
holds for some $i$ and whenever $Committed(M,S_j)$ holds for some $M
\in \mathcal{M}$, then $i \le j$ holds. If there is more than one
machine $M_1$ with this property, we randomly choose one of them.

Take the run of $TA(\{ M_1 \},\ASM{TaCtl})$ starting in state $S_0$,
say $S_0, S_1^\prime, S_2^\prime, \dots, S_n^\prime$. As $M_1$
commits, this run is finite. $M_1$ has been $\ASM{Delete}$d from
$TransAct$ and none of the $\ASM{TaCtl}$ components is triggered any
more: neither $\ASM{Commit}$ nor $\ASM{LockHandler}$ because
$CommitRequest$ resp. $LockRequest$ remain empty; not
$\ASM{DeadlockHandler}$ because $Deadlock$ remains false since~$M_1$
never $Wait$s for any machine; not $\ASM{Recovery}$ because$Victim$
remains empty. Note that in this run the schedule for $M_1$ is already
cleansed.

We now define a run $S_0^{\prime\prime} , S_1^{\prime\prime},
S_2^{\prime\prime}, \dots$ (of $TA({\cal M} - \{ M_1 \},\ASM{TaCtl})$,
as has to be shown) which starts in the final state $S_n^\prime =
S_0^{\prime\prime}$ of the $TA(\{ M_1 \},\ASM{TaCtl})$ run and where
we remove from the run defined by the cleansed schedules $\Delta_i(M)$
for the originally given run all updates made by steps of~$M_1$ and
all updates in $\ASM{TaCtl}$ steps which concern~$M_1$. Let
\[ \Delta_i^{\prime\prime} = \bigcup\limits_{M \in \mathcal{M} - \{M_1\} } \Delta_i(M) \cup \{ (\ell,v) \in \Delta_i(\ASM{TaCtl}) \mid (\ell,v) \;\text{does not concern $M_1$} \} . \]

That is, in the update set $\Delta_i^{\prime\prime}$ all updates are removed from the
original run which are done by $M_1$---their effect is reflected
already in the initial run segment from $S_0$ to $S_n^\prime $---or
are $\ASM{LockHandler}$ updates involving a $LockRequest(M_1,L)$ or
are $Victim(M_1):=true$ updates of the $\ASM{DeadlockHandler}$ or are
updates involving a $\ASM{TryToRecover}(M_1)$ step or are done by
a step involving a $\ASM{Commit}(M_1)$.

\begin{lemma}\label{lem1}

\ $S_0^{\prime\prime}, S_1^{\prime\prime}, S_2^{\prime\prime}, \dots$
is a run of $TA({\cal M} - \{ M_1 \},\ASM{TaCtl})$.

\end{lemma}

\begin{lemma}\label{lem2}

\ The run $S_0, S_1^\prime, S_2^\prime, \dots, S_n^\prime,
S_1^{\prime\prime}, S_2^{\prime\prime}, \dots$ of $TA({\cal
  M},\ASM{TaCtl})$ is equivalent to the original run
$S_0,S_1,S_2,\dots$.

\end{lemma}

By induction hypothesis $S_0^{\prime\prime}, S_1^{\prime\prime},
S_2^{\prime\prime}, \dots$ is serialisable, so
$S_0,S_1^\prime,S_2^\prime,\dots$ and thereby also $S_0,S_1,S_2,\dots$ is
serialisable with $M_1 < M$ for all $M \in \mathcal{M} - \{ M_1 \}$.\hfill
\end{proof}

\begin{proof}\textbf{(Lemma \ref{lem1})}
\ We first show that omitting in $\Delta_i^{\prime\prime}$ every update
from $\Delta_i(\ASM{TaCtl})$ which concerns $M_1$ does not affect
updates by $\ASM{TaCtl}$ in $S_i^{\prime\prime}$ concerning $M \neq
M_1$. In fact starting in the final $M_1$-state $S_0^{\prime\prime}$,
$TA({\cal M} - \{ M_1 \},\ASM{TaCtl})$ makes no move with a
$Victim(M_1):=true$ update and no move of $\ASM{Commit}(M_1)$ or
$\ASM{HandleLockRequest}(M_1,L)$ or $\ASM{TryToRecover}(M_1)$

It remains to show that every $M$-step defined by
$\Delta_i^{\prime\prime}(M)$ is a possible $M$-step in a $TA({\cal M}
- \{ M_1 \},\ASM{TaCtl})$ run starting in $S_0''$.  Since the
considered $M$-schedule $\Delta_i(M)$ is cleansed, we only have to
consider any proper update step of $M$ in state~$S_i''$ (together with
its preceding lock request step, if any). If in~$S_i''$ $M$ uses
$newLocks$, in the run by the cleansed schedules for the original run
the locks must have been granted after the first $\ASM{Commit}$, which
is done for~$M_1$ before $S_0''$. Thus these locks are granted also
in~$S_i^{\prime\prime}$ as part of a $TA({\cal M} - \{ M_1
\},\ASM{TaCtl})$ run step. If no $newLocks$ are needed, that proper
$M$-step depends only on steps computed after $S_0''$ and thus is part
of a $TA({\cal M} - \{ M_1 \},\ASM{TaCtl})$ run step. \hfill
\end{proof}

\begin{proof}\textbf{(Lemma \ref{lem2})}
\ The cleansed machine schedules in the two runs, the read locations
and the values read there have to be shown to be the same.  First
consider any $M \not = M_1$. Since in the initial segment $S_0,
S_1^\prime, S_2^\prime, \dots, S_n^\prime$ no such~$M$ makes any move
so that its update sets in this computation segment are empty, in the
cleansed schedule of~$M$ for the run $S_0, S_1^\prime, S_2^\prime,
\dots, S_n^\prime, S_1^{\prime\prime}, S_2^{\prime\prime}, \dots$ all
these empty update sets disappear. Thus this cleansed schedule is the
same as the cleansed schedule of $M$ for the run
$S_n^\prime,S_1^{\prime\prime}, S_2^{\prime\prime}, \dots$ and
therefore by definition of $\Delta_i^{\prime\prime}(M) = \Delta_i(M)$
also for the original run $S_0,S_1,S_2,\dots$ with same read locations
and same values read there.

Now consider $M_1$, its schedule $\Delta_0(M_1), \Delta_1(M_1), \dots$
for the run $S_0,S_1,S_2,\dots$ and the corresponding cleansed
schedule $\Delta_{i_0}(M_1), \Delta_{i_1}(M_1), \Delta_{i_2}(M_1),
\dots$. We proceed by induction on the cleansed schedule steps of
$M_1$.  When $M_1$ makes its first step using the
$\Delta_{i_0}(M_1)$-updates, this can only be a proper~$M_1$-step
together with the corresponding $\ASM{Record}$ updates (or a lock
request directly preceding such a $\Delta_{i_1}(M_1)$-step) because in
the computation with cleansed schedule each lock request of $M_1$ is
granted and $M_1$ is not $Victim$ized. The values $M_1$ reads or
writes in this step (in private or locked locations) have not been
affected by a preceding step of any $M \not = M_1$---otherwise~$M$
would have locked before the non-private locations and keep the locks
until it commits (since cleansed schedules are without $\ASM{Undo}$ steps),
preventing $M_1$ from getting these locks which contradicts the fact
that $M_1$ is the first machine to commit and thus the first one to
get the locks. Therefore the values $M_1$ reads or writes in the step
defined by $\Delta_{i_0}(M_1)$ (resp. also $\Delta_{i_1}(M_1)$)
coincide with the corresponding location values in the first
(resp. also second) step of $M_1$ following the cleansed schedule to
pass from $S_0$ to $S_1^\prime$ (case without request of $newLocks$)
resp. from $S_0$ to $S_1^\prime$ to $S_2^\prime$ (otherwise). The same
argument applies in the inductive step which establishes the claim.\hfill
\end{proof}

\section{Conclusion}

In this article we specified (in terms of Abstract State Machines) a
transaction controller $\ASM{TaCtl}$ and a transaction operator which
turn the behaviour of a set of concurrent programs into a
transactional one under the control of $\ASM{TaCtl}$. In this way the
locations shared by the programs are accessed in a well-defined
manner. For this we proved that all concurrent transactional runs are
serialisable.

The relevance of the transaction operator is that it permits to
concentrate on the specification of program behavior ignoring any
problems resulting from the use of shared locations. That is,
specifications can be written in a way that shared locations are
treated as if they were exclusively used by a single program. This is
valuable for numerous applications, as shared locations (in
particular, locations in a database) are common, and random access to
them is hardly ever permitted.

Furthermore, by shifting transaction control into the rigorous
framework of Abstract State Machines we made several extensions to
transaction control as known from the area of databases
\cite{elmasri:2006}. In the classical theory schedules are sequences
containing read- and write-operations of the transactions plus the
corresponding read- and write-lock and commit events, i.e., only one
such operation or event is treated at a time. In our case we exploited
the inherent parallelism in ASM runs, so we always considered an
arbitrary update set with usually many updates at the same time. Under
these circumstances we generalised the notion of schedule and
serialisability in terms of the synchronous parallelism of ASMs. In
this way we stimulate also more parallelism in transactional systems.

Among further work we would like to be undertaken is to provide a
(proven to be correct) implementation of our transaction controller
and the $TA$ operator, in particular as plug-in for the
CoreASM~\cite{farahbod:scp2014,farahbod:fi2007} or Asmeta~\cite{arcaini:spe2011,gargantini:jucs2008}
simulation engines. We would also like to see refinements or
adaptations of our transaction controller model for different
approaches to serialisability~\cite{gray:1993}, see also the ASM-based
treatment of multi-level transaction control
in~\cite{kirchberg:2009}. Last but not least we would like to see
further detailings of our correctness proof to a mechanically verified
one, e.g. using the ASM theories developed in KIV (see~\cite{kiv} for
an extensive list of relevant publications) and
PVS~\cite{GarRic00,GoVhLa96,Verifix96} or the (Event-)B~\cite{Abrial96,Abrial10} theorem prover for an (Event-)B transformation of
$TA({\cal M},\ASM{TaCtl})$ (as suggested in~\cite{GlaHLR13}).

\subsubsection{Acknowledgement.} 
We thank Andrea Canciani and some of our referees for useful comments
to improve the paper.

\def\note#1{}

\bibliographystyle{abbrv}

\bibliography{TransactionAsmbib}

\end{document}